\documentclass[conference]{IEEEtran}
\usepackage{}
\usepackage{mathrsfs,amsmath,amssymb,epsfig,amsfonts,fancyhdr,graphicx,cite,amsthm}

\newtheorem{thm}{Theorem}[]
\newtheorem{cor}{Corollary}[]

\theoremstyle{remark}

\theoremstyle{definition}

\begin{document}

\title{A Partial Decode-Forward Scheme For A Network with $N$ relays}
\author{
Yao Tang\\
\authorblockA{ECE Department, McGill University\\
Montreal, QC, Canada\\
Email: yao.tang2@mail.mcgill.ca}
\and
Mai Vu\\
\authorblockA{
ECE Department, Tufts University\\
Medford, MA, USA\\
Email: maivu@ece.tufts.edu}
}
\maketitle

\begin{abstract}
We study a discrete-memoryless relay network consisting of one source, one
destination and $N$ relays, and design a scheme based on partial
decode-forward relaying. The source splits its message into one common and
$N+1$ private parts, one intended for each relay. It encodes these message
parts using $N$th-order block Markov coding, in which each private message
part is independently superimposed on the common parts of the current and
$N$ previous blocks. Using simultaneous sliding window decoding, each relay
fully recovers the common message and its intended private message with the
same block index, then forwards them to the following nodes in the next
block. This scheme can be applied to any network topology. We derive its
achievable rate in a compact form. The result reduces to a known
decode-forward lower bound for an N-relay network and partial
decode-forward lower bound for a two-level relay network. We then apply the
scheme to a Gaussian two-level relay network and obtain its capacity lower
bound considering power constraints at the transmitting nodes.
\end{abstract}

\IEEEpeerreviewmaketitle
%----------------------------------------------------
\section{Introduction}\label{sec:intro}
The relay channel first introduced by van der Meulen \cite{meulen1971three} consists of a source aiming to communicate with a destination with the help of a relay. In \cite{cover1979capacity}, Cover and El Gamal introduce the cut-set bound and two coding strategies, namely decode-forward and compress-forward, for the basic three-node relay channel. By allowing the relay to decode only a part of the transmitted message, partial decode-forward can be considered as the generalization of decode-forward \cite{cover1979capacity}, \cite{arefPDF}.

In the past few years, substantial research activities have been dedicated to extending the classical one-relay channel to a general relay network consisting of $N$ communicating parties. In \cite{gastpar2005large_gaussian_RN}, Gastpar and Vetterli discuss the asymptotic capacity in the limit as the number of relays tends to infinity and the scaling behavior of capacity for a large class of Gaussian relay networks. Recently, Lim, Kim, El Gamal and Chung propose a compress-forward based scheme (noisy network coding) \cite{sung2011noisy} for the general multi-source multicast noisy network, which includes network coding \cite{ahlswede2000networkInfoFlow} as a special case.\par

 However, it is still unclear how to generalize decode-forward relaying to the multi-source multicast network. In \cite{Xie2005multilevelrelay}, Xie and Kumar analyze a multiple-level relay channel with one source and one destination and give an achievable rate based on full decode-forward. This scheme is extended in \cite{aref2007pdfnetwork}, in which all relays successively decode only part of the messages of the previous relay, and obtains the capacity of semi-deterministic and orthogonal relay networks. In \cite{aref2009pdftwolevel} and \cite{aref2010pdftwolevel}, Ghabeli and Aref generalize partial decode-forward in a two-level relay network, considering all possible partial decoding conditions that can occur among different message parts of at the source and the relays.

There are three common approaches for the decode-forward strategy, namely: (a) irregular encoding/sequential decoding; (b) regular encoding/simultaneous sliding window decoding; (c) regular encoding/backward decoding \cite{Kramer2005cooperative}. For semi-deterministic relay networks \cite{aref2011SSRN}, the second and the third approaches can achieve the same rate, which is greater than that of the first approach. Furthermore, the second approach creates less delay than the third one.\par
In this paper, we propose a novel transmission scheme for a single-source single-destination network with $N$ relays based on regular encoding and simultaneous sliding window decoding. The source splits its message into one common and $N+1$ private parts  and performs block Markov coding. Each relay helps forward the common part and the private message part intended for itself. We derive the achievable rate in a compact form and show that this scheme can reduce to the network decode-forward scheme of \cite{cover1979capacity} and partial decode-forward for two-level relay network in \cite{aref2010pdftwolevel}. Finally, we analyze a two-level relay network in AWGN environments and provide the achievable rate.\par
\begin{figure}[t]
\center
  \includegraphics[scale=0.65]{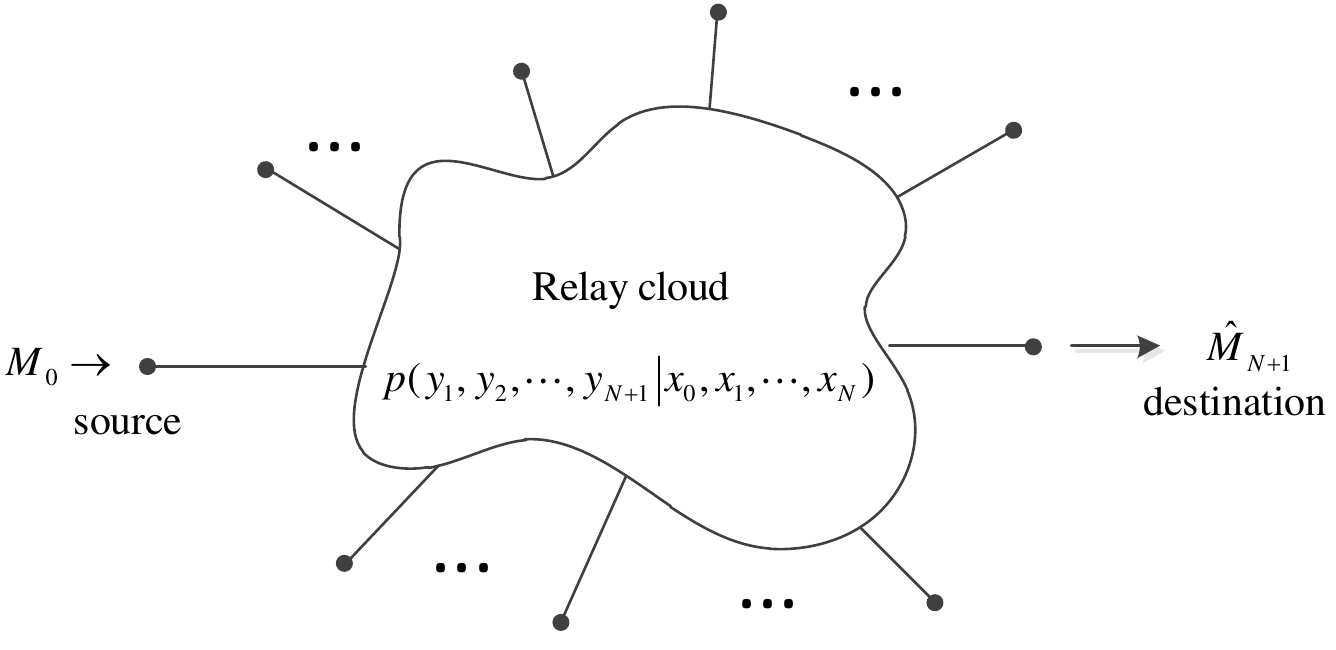}
  \caption{General discrete memoryless relay network}\label{p4}
\end{figure}
%----------------------------------------------------------------
\section{Preliminaries}\label{sec_pre}
\subsection{Discrete Memoryless Relay Networks}
Consider a discrete memoryless relay network (DM-RN) with $N+2$ nodes ($\mathcal{X}_0\times\mathcal {X}_1\times \dots \times \mathcal {X}_{N}$, $p(y_1,y_2,\dots,y_{N+1}|x_0,x_1,\dots,x_{N})$, $\mathcal{Y}_1\times\mathcal{Y}_2\times \dots \times \mathcal {Y}_{N+1}$), where source node 0 wants to send a message $M$ to the destination node $N+1$ with the help of relay nodes $1,\dots,N$, as shown in Figure \ref{p4}. A $(2^{nR},n)$ code for this DM-RN consists of:
\begin{itemize}
\item
A message set $\mathcal{M}=[1:2^{nR}].$
\item
A source encoder that assigns a codeword $x_0^n(m)$ to each message $m\in [1:2^{nR}]$.
\item
A set of relay encoders $k\in[1:N]$, which assigns a symbol $x_{ki}(y_{k}^{i-1})$ to each received sequence $y_{k}^{i-1}$ for $i\in[1:n]$.
\item
A destination decoder, which assigns an estimate $\hat{m}_{N+1}$ to each received sequence $y_{N+1}^n \in \mathcal{Y}_{N+1}^n$, or declares an error message $e$.
\end{itemize}\par
Definitions for the average error probability, achievable rate and capacity follow the standard ones in \cite{gamal2010lecture}.
\subsection{Definitions and Notation}
To make the following analysis more concise and readable, we introduce some definitions and clarify notation in this section.
\begin{itemize}
\item
Define $T=\{1,\dots,N\}$ as the complete set of all relays.
\item
Define $S$ to be a subset of $T$, that is $S\subseteq T$ and $S^c=T-S$. Either $S$ or $S^c$ can be empty and the largest $S$ is $T$.
%%Define $S_j=S-\{j\}$ and $S_j^c=S^c-\{j\}$.\par
\item
Denote $M_{i}^{j}=\{M_{i},M_{i+1},\dots,M_{j}\}$, where $j\geq i$. For example, $U_2^{N-1}$ means $\{U_2,U_3,\dots,U_{N-1}\}$, where $N\geq 3$.
\item
Given nonempty set $L$ and variable $M$, let $M_L=\{M_i\}_{i\in L}=\{M_a,M_b,M_c,\dots\}$, where $a,b,c,\dots\in L$ and are different from each other. $|L|$ signifies the cardinality of $L$.
\end{itemize}
\section{A Network Partial Decode-Forward Scheme}\label{sec_network_pdf}
\begin{figure}
\center
  \includegraphics[scale=0.65]{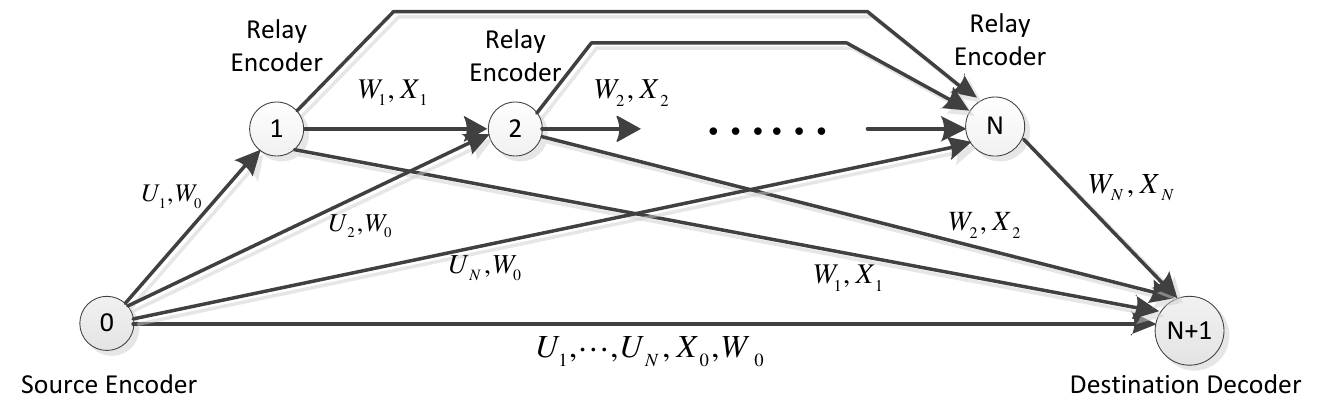}
  \caption{The proposed private message scheme for a single-source single-destination network with $N$ relays} \label{pdf_RN_1}
\end{figure}
Consider a network consisting of a single source, single destination and $N$ relays as in Figure \ref{pdf_RN_1}. The source has direct links to all relays and to the destination, and connections among the relays and the destination are arbitrary. We design a novel coding scheme for this relay network based on partial decode-forward relaying.

Let all relay nodes be ordered in an arbitrary order permutation $\pi(\cdot)$. In each order, we assume that the $k$th relay decodes information from all nodes below it, (i.e. order $\{1,\dots,k-1\}$) and forwards information to nodes above it (i.e. order $\{k+1,\dots,N\}$). Next, we will describe the scheme for the nominal order $\pi=[1,2,\dots,N]$ to simplify the notation, keeping in mind that it can also be applied to any other order $\pi$.

The new idea in this scheme is the way it performs rate splitting. At each block transmission, the source splits its message into $N+2$ parts: a common message and $N+1$ private messages, one intended for each relay and one for the destination. These messages are then encoded using $N$th order block Markov coding. Each relay fully recovers the common message and its intended private message with same block index as the common message, then forwards them together in the next block.

Specifically, let the source message in block $j$ be split as $m_j=(m_{0j},m_{1j},\dots,m_{(N+1)j})$, where $m_0$ denotes the common message that is forwarded among all relays, $m_k$ denotes the message intended to be decoded at relay $k$, but not at other relays, and $m_{N+1}$ denotes the message intended to be decoded only at the destination. The rate is $R=\sum_{i=0}^{N+1}R_i$. \par
Block Markov superposition coding is used to generate the independent codewords in each block as follows (for simpler notation we suppress block index in codewords here, but will include it in the detailed proof later).
\begin{itemize}
\item
$W_k, k\in \{0\}\cup T$, carries common message $m_{0,j-k}$ of different blocks. $W_k$ are successively superimposed on each other as in block Markov encoding.
\item
$\{U_k\}, {k\in T}$, carries private  message $m_k$ to be decoded at relay $k$ and not decoded at other relays. Each $U_k$ is superimposed on all $W_k$.
\item
$X_k, k\in T$, is the codeword sent by relay $k$ which supports the forwarding of the message in $U_k$ (of the previous block) and all $W_l$ ($l\leq k$).
\item
$X_0$ is the codeword sent by the source which carries all messages including the remaining message $m_{N+1}$ to be decoded only at the destination. $X_0$ is superimposed on all $\{W_0^N\}$, $\{U_1^N\}$ and $\{X_1^N\}$.
\end{itemize}\par
At each block $j$, the source sends $X_0$ which contains all $\{W_0^N\}$, $\{U_1^N\}$ and $\{X_1^N\}$. Each relay $k$ decodes $\{W_0,\dots,W_{k-1}\}$ from all previous nodes and $U_k$ from the source. Then in block $j+1$, it transmits $X_k$ which carries its private message of block $j-k$ $(m_{k,j-k})$ superimposed on all previous-block common messages. The destination uses joint decoding simultaneously over all blocks. Specifically, it waits until the end of the last block to decode all messages carried by $\{W_0^N\}$, $\{U_1^N\}$ and $\{X_0^N\}$ simultaneously using signals received in the last $N$ blocks. This coding scheme is illustrated in Figure \ref{pdf_RN_1}.
\begin{thm}
For a single-source single-destination network with $N$ relays ($\mathcal{X}_0\times\mathcal {X}_1\times \dots \times \mathcal {X}_{N}$, $p(y_1,y_2,\dots,y_{N+1}|x_0,x_1,\dots,x_{N})$, $\mathcal{Y}_1\times\mathcal{Y}_2\times \dots \times \mathcal {Y}_{N+1}$), by using partial decode-forward, the capacity $C$ is lower bounded by (\ref{thm_pdf_DM-RN}), where $\pi(\cdot)$ denotes a permutation order for the relay nodes.
\begin{figure*}
\begin{align}
\label{thm_pdf_DM-RN}
C\geq\sup_{\pi(T)}\sup_{p_\pi}
\min_{ S\subset \pi(T)}
 \min\left\{
\begin{array}{ll}
I\left(U_{1}^{N},X_{0}^{N},W_{0}^{N};Y_{N+1} \right),\\
I\left(X_{0},X_{S},U_{S};Y_{N+1}|X_{S^c},U_{S^c},W_{0}^{N}\right)
+\min_{\pi_j\in S^c}I\left(W_0,W_{\pi_1}^{\pi_{j-1}};Y_{\pi_{j}}|X_{\pi_j},W_{\pi_j}^{\pi_N}\right)\\
+\sum_{\pi_j\in S^c} I\left(U_{\pi_j};Y_{\pi_j}|W_{0}^{N},X_{\pi_j}\right)
\end{array}
\right.
\end{align}
where $T=\{1,2,\dots,N\}$, $\pi(T)$ is a permutation order of $T$, $S$ is a subset of $\pi(T)$, and
\begin{align}
\label{coding_distribution}
 p_{\pi}=p(U_1^{N},W_0^{N},X_0^{N})
 = p(X_0|U_1^{N},W_0^{N},X_1^{N})
\prod_{\pi_k =\pi_1}^{\pi_N}p(U_{\pi_k}|W_0^{N},X_{\pi_{k}})p(W_0|W_1^{N})
\prod_{\pi_k =\pi_1}^{\pi_N}p(X_{\pi_k}|W_{\pi_k}^{\pi_N})p(W_{\pi_k}|W_{\pi_{k+1}}^{\pi_N}).
\end{align}
\end{figure*}
\end{thm}
%\begin{figure*}
%\begin{align}
%\label{thm_pdf_DM-RN}
%C\geq\sup_{p}
%\min\left\{
%\begin{array}{ll}
%I\left(U_2^{N-1},X_1^{N-1},W_1^{N-1};Y_N \right),\\
%\left\{I\left(X_1;Y_N\left| U_2^{N-1},X_2^{N-1},W_1^{N-1}\right. \right)+I\left(U_j,W_1^{j-1};Y_j \left| X_j,W_j^{N-1}\right.\right)+\sum_{i\in T_j} I\left(U_i;Y_i \left| W_1^{N-1},X_i\right.\right)\right\}_{j\in T},\\
%\left\{I\left(X_1,X_{T_j},U_{T_j};Y_N \left| X_j,U_j,W_1^{N-1}\right.\right)+I\left(U_j,W_1^{j-1};Y_j\left|X_j,W_j^{N-1}\right.\right)\right\}_{j\in T},\\
%\min_s\left\{\right.I\left(X_1,X_{S_j},U_{S_j};Y_N \left| X_{S^c},U_{S^c},X_j,U_j,W_1^{N-1}\right.\right)+I\left(U_j,W_1^{j-1};Y_j \left| X_j,W_j^{N-1}\right.\right)\\
%~+\sum_{i\in S^c_j} I\left(U_i;Y_i \left| W_1^{N-1},X_i\right.\right)\left.\right\}_{j\in T}
%\end{array}
%\right.
%\end{align}
%\end{figure*}
%\end{thm}
\begin{proof}
We use a block coding scheme in which each user sends $b-N$ messages over $b$ blocks of $n$ symbols. Each relay and the destination employ simultaneous decoding.\par
\subsection{Codebook generation}
\begin{figure*}
\center
  \includegraphics[scale=0.74]{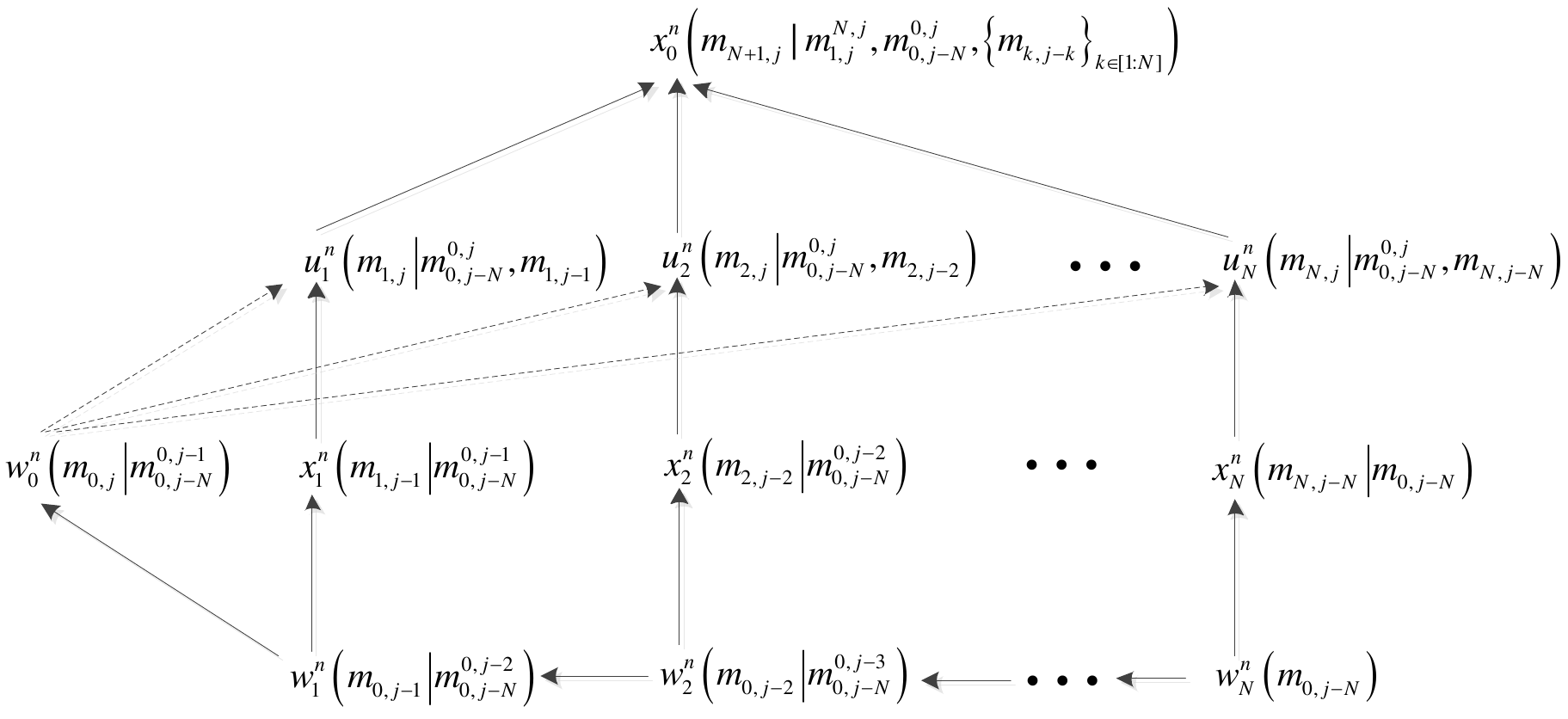}
  \caption{Encoding diagram of a single-source single-destination network with $N$ relays (arrows denote superposition structure)} \label{encoding_diagram}
\end{figure*}
Fix the joint distribution in (\ref{coding_distribution}) where the meaning of each component is as follows: \par %$p(U_2^{N-1},W_1^{N-1},X_1^{N-1})=\\ p(X_1|U_2^{N-1},W_1^{N-1},X_2^{N-1})
%\prod_{k=2}^{N-1}p(U_k|W_1^{N-1},X_k)\\
%p(W_1|W_2^{N-1})
%\prod_{k=2}^{N-1}p(X_k|W_k^{N-1})p(W_k|W_{k+1}^{N-1})$
\begin{itemize}
\item
$p(W_k|W_{k+1}^{N})$: at relay $k$, the current common message is superimposed on previous common messages.
\item
$p(X_k|W_k^{N})$: at relay $k$, the current private message is superimposed on the current and previous common messages.
\item
$p(W_0|W_1^{N})$: at the source, the current common message is superimposed on the common messages of all previous $N$ blocks.
\item
$p(U_k|W_0^{N},X_k)$: at the source, the private message for a specific relay is superimposed on all common messages and the previous private message for that relay.
\item
$p(X_0|U_1^{N},W_0^{N},X_1^{N})$: at the source, the private message for the destination is superimposed on all common messages and all private messages for all relays.
\end{itemize}\par
Figure \ref{encoding_diagram} illustrates the superposition coding structure. In block $j$, the source splits its message as $m_j=[m_{N+1,j},m_{N,j},\dots,m_{1,j},m_{0,j}]$. For every relay node $k=N,\dots,1$ and every message set $\{m_{0,j-N},\dots,m_{0,j-k}, m_{k,j-k}\}$:
\begin{itemize}
\item
Randomly and independently generate $2^{nR_0}$ sequences $w_k^n(m_{0,j-k}|m_{0,j-N}^{0,j-k-1})$ for all $m_{0,j-k}\in [1:2^{nR_0}]$, each according to $\prod_{i=1}^np_{W_k|W_{k+1}^{N}}(w_{ki}|w_{k+1,i}^{N,i})$,
\item
Randomly and independently generate $2^{nR_k}$ sequences $x_k^n(m_{k,j-k}|m_{0,j-N}^{0,j-k})$ for all $m_{k,j-k}\in [1:2^{nR_k}]$, each according to $\prod_{i=1}^np_{X_k|W_{k}^{N}}(x_{ki}|w_{k,i}^{N,i})$.
\end{itemize}
For source node $k=0$ and each message set $m_{0,j}^{N+1,j}$.
\begin{itemize}
\item
For all sequence $w_k^n(m_{0,j-k}|m_{0,j-N}^{0,j-k-1})$ with $k\in T$, randomly and independently generate $2^{nR_0}$ sequences $w_0^n(m_{0,j}|m_{0,j-N}^{0,j-1})$ for $m_{0,j}\in [1:2^{nR_0}]$, each according to $\prod_{i=1}^np_{W_0|W_1^{N}}(w_{0,i}|w_{1,i}^{N,i})$,
\item
For each $k\in T$, randomly and independently generate $2^{nR_k}$ sequences $u_k^n(m_{k,j}|m_{0,j-N}^{0,j},m_{k,j-k})$ for all $m_{k,j-k}\in [1:2^{nR_k}]$, each according to $\prod_{i=1}^np_{U_k|W_{0}^{N},X_k}(u_{ki}|w_{0,i}^{N,i},x_{k,i})$,
\item
Randomly and independently generate $2^{nR_{N+1}}$ sequences $x_0^n(m_{N+1,j}|m_{1,j}^{N,j},m_{0,j-N}^{0,j},\left\{m_{k,j-k}\right\}_{k\in T})$ for all $m_{N+1,j}\in [1:2^{nR_{N+1}}]$, each according to \par $\prod_{i=1}^np_{X_0|U_1^{N},W_{0}^{N},X_1^{N}}(x_{0i}|u_{1,i}^{N,i},w_{0,i}^{N,i},x_{1,i}^{N,i})$.
\end{itemize}\par
The codebook is independently generated in each block as above and then is revealed to all the parties.
\subsection{Encoding}
To send $\{m_{N+1,j},\dots,m_{0j}\}$ in block $j$, the source transmits $x_0^n(m_{N+1,j}|m_{1,j}^{N,j},m_{0,j-N}^{0,j},\{m_{k,j-k}\}_{k\in T})$ from codebook $\mathcal{C}_j$. At the end of block $j$, each relay $k\in T$ has an estimate $\tilde{m}_{k,j-k+1}$ of message $m_{k,j-k+1}$ and $\tilde{m}_{0,j-k+1}$ of message $m_{0,j-k+1}$. In the block $j+1$, each relay $k\in T$ transmits $x_k^n(\tilde{m}_{k,j-k+1},\tilde{m}_{0,j-N+1}^{0,j-k+1})$ from codebook $\mathcal{C}_{j+1}$.
\subsection{Decoding}
\subsubsection{Simultaneous decoding at the first relay}
At the end of block $j$, by knowing $m_{0,j-N}^{0,j-1}$ and $m_{2,j-1}$, the first relay $k=1$ decodes $m_{1,j}$ and $m_{0,j}$ such that:
\begin{align}
\nonumber
\left(
u_1^n(m_{1,j}|m_{0,j-N+1}^{0,j},m_{1,j-1}),
w_0^n(m_{0,j}|m_{0,j-N+1}^{0,j-1}),\right.\\
\nonumber
w_1^n(m_{0,j-1}|m_{0,j-N}^{0,j-2}),
x_1^n(m_{1,j-1}|m_{0,j-N}^{0,j-1}),\\
\left.w_2^n,w_3^n,\dots,w_{N}^n,
y_1^n(j)\right)
\in T_\epsilon^{(n)}.
\end{align}
The decoding error probability goes to 0 as $n\to\infty$, if
\begin{align}
R_1&<I(U_1;Y_2|W_0^{N},X_1),\\
R_1+R_0&<I(U_1,W_0;Y_2|X_1,W_1^{N}).
\end{align}
\subsubsection{Simultaneous sliding window decoding at other relays $k\in[2:N]$}
At the end of block $j$, by knowing $m_{0,j-N-k+1}^{0,j-k}$ and $m_{k,j-k}$, the relay node $k$ will decode $m_{k,j-k+1}$ and $m_{0,j-k+1}$ such that the following conditions hold simultaneously:
\begin{align*}
\nonumber
\left(w_{k-1}^n(m_{0,j-k+1}|m_{0,j-N}^{0,j-k}),w_k^n(m_{0,j-k}|m_{0,j-N}^{0,j-k-1}),
\right.
\\
\nonumber
\left.x_k^n(m_{k,j-k}|m_{0,j-N}^{0,j-k}),
w_{k+1}^n,w_{k+2}^n,\dots,w_{N}^n,
y_k^n(j)
\right)
\in T_\epsilon^{(n)}.\\
%\end{align*}
%\begin{align*}
%\left\{
%\begin{array}{ll}
%\nonumber
%w_{k-2}^n(m_{0,j-k+1}|m_{0,j-N-1}^{0,j-k}),\\
%\nonumber
%w_{k-1}^n(m_{0,j-k}|m_{0,j-N-1}^{0,j-k-1}),
%w_k^n(m_{0,j-k-1}|m_{0,j-N-1}^{0,j-k-2}),\\
%x_k^n(m_{k,j-k-1}|m_{0,j-N-1}^{0,j-k-1}),\\
%\nonumber
%w_{k+1}^n,w_{k+2}^n,\dots,w_{N}^n,
%y_k^n(j-1)
%\end{array}
%\right\}
%\in T_\epsilon^{(n)}.\\
\vdots~~~~~~~~~~~~~~~~~~~~~~~~~~~~~~~~~~\\
\nonumber
\left(w_1^n(m_{0,j-k+1}|m^{0,j-k}_{0,j-k-N+2}),
\nonumber
w_2^n,\dots,w_{k-1}^n,w_k^n,
x_k^n,\right.\\
\left.
w_{k+1}^n,\dots,w_{N}^n,
y_k^n(j-k+2)
\right)
\in T_\epsilon^{(n)}.
\end{align*}
\begin{align}
\nonumber
\left(
u_k^n(m_{k,j-k+1}|m_{0,j-N-k+1}^{0,j-k+1},m_{k,j-2k+1}),\right.\\
\nonumber
w_0^n(m_{0,j-k+1}| m_{u,j-N-k+1}^{0,j-k}),
w_1^n,w_2^n,\dots,
w_{k-1}^n,\\
\left.
x_k^n,w_k^n,
 w_{k+1}^n,\dots,w_{N}^n,
y_k^n(j-k+1)
\right)
\in T_\epsilon^{(n)}.
\end{align}
Therefore, there are $k$ decoding rules to be satisfied simultaneously at relay $k$. The decoding error probability goes to 0, as $n\to\infty$, if
\begin{align}
\label{combine_3}
R_k&<I(U_k;Y_k|W_0^{N},X_k),\\
\label{combine_4}
R_k+R_0&<I(U_k,W_0^{k-1};Y_k|X_k,W_k^{N}).
\end{align}
Detailed error analysis at relay $k$ is in Appendix \ref{error_analysis_relay}.
\subsubsection{Simultaneous sliding window decoding at destination node $N+1$}
At the end of block $j$, the destination node $N+1$ will decode $m_{k,j-N}$ for all $k\in T\cup\{N+1\}$ and $m_{0,j-N}$ such that the following conditions hold simultaneously:
\begin{align*}{}
\left(
\nonumber
x_{N}^n(m_{N,j-N}|m_{0,j-N}),
w_{N}^n(m_{0,j-N}),
y_{N+1}^n(j)\right)
\in T_\epsilon^{(n)}.\\
\nonumber
\vdots~~~~~~~~~~~~~~~~~~~~~~~~~~~~~~~~~~~~~~~~~~~\\
\left(
\nonumber
x_{N-i+1}^n(m_{N-i+1,j-N}|m_{0,j-N+1-i}^{0,j-N}),\right.\\
\nonumber
w_{N-i+1}^n(m_{0,j-N}|m_{0,j-N+1-i}^{0,j-N}),\\
\left.x_{N-i+1}^n,w_{N-i+1}^n,\dots,x_{N}^n,w_{N}^n,
y_{N+1}^n(j-i+1)
\right)
\in T_\epsilon^{(n)}.\\
\nonumber
\vdots~~~~~~~~~~~~~~~~~~~~~~~~~~~~~~~~~~~~~~~~~~~\\
\left(
\nonumber
x_1^n(m_{1,j-N}|m_{0,j-2N+1}^{0,j-N}),
w_1^n(m_{0,j-N}|m_{0,j-2N+1}^{0,j-N-1}),\right.\\
\left.x_2^n,w_2^n,\dots,x_{N}^n,w_{N}^n,
y_{N+1}^n(j-N+1)
\right)
\in T_\epsilon^{(n)}.
\end{align*}
And,
\begin{align}
\nonumber
\left(
\{u_k^n(m_{k,j-N}|m_{0,j-2N}^{0,j-N},m_{k,j-N-k})\}_{k \in T},\right.\\
\nonumber
x_0^n(m_{N+1,j-N}|m_{1,j-N}^{N,j-N},
m_{0,j-2N}^{0,j-N},\\
\nonumber
     ~~~ \{m_{k,j-N-k}\}_{k \in T}),
w_0^n(m_{0,j-N}|m_{0,j-2N}^{0,j-N-1}),\\
\left.
x_1^n,w_1^n,x_2^n,w_2^n,\dots,x_{N}^n,w_{N}^n,
y_{N+1}^n(j-N)
\right)
\in T_\epsilon^{(n)}.
\end{align}
We have $N+1$ decoding rules to be satisfied simultaneously. The decoding error probability goes to 0, as $n\to\infty$, if
\begin{align}
\label{combine}
\sum_{i=0}^{N+1}R_i&<I(U_1^{N},X_0^{N},W_0^{N};Y_{N+1}),\\
\label{combine_2}
\sum_{i\in S} R_i+R_{N+1}&<I(X_0,X_S,U_S;Y_N|X_{S^c},U_{S^c},W_0^{N}),\\
\label{combiniation_5}
R_{N+1}&<I(X_0;Y_{N+1}|U_1^{N},X_1^{N},W_0^{N}),\\
\label{combiniation_6}
\sum_{i=1}^{N+1}R_i&<I(U_1^{N},X_0^{N};Y_{N+1}|W_0^{N}),
\end{align}
where $S$ is a subset of $T$ that contains wrongly decoded messages at the relay. Detailed error analysis at the destination is in Appendix \ref{error_analysis_destination}.
%If $S=\emptyset$ in (\ref{combine_2}), we have
%\begin{align}
%\label{combine_1}
%R_N<I(X_1;Y_N|U_2^{N-1},X_2^{N-1},W_1^{N-1}),
%\end{align}
%If $S=T$ in (\ref{combine_2}), we have
%\begin{align}
%\sum_{i=2}^{N}R_i<I(U_2^{N-1},X_1^{N-1};Y_N|W_1^{N-1}).
%\end{align}

\subsection{Combination Process}
From (\ref{combine}), we can directly get that:
\begin{align}
R<I(U_1^{N},X_0^{N},W_0^{N};Y_{N+1}).
\end{align}\par
%
%From (\ref{combine_1}), (\ref{combine_3}) and (\ref{combine_4}), we get,
%\begin{align*}
%R&=\{R_N+(R_j+R_u)+\sum_{i\in T,i\neq j} R_i\}_{j\in T}\\
%&<\{I(X_1;Y_N|U_2^{N-1},X_2^{N-1},W_1^{N-1})\\
%&~~~+I(U_j,W_1^{j-1};Y_j|X_j,W_j^{N-1})\\
%&~~~+\sum_{i\in T,i\neq j} I(U_i;Y_i|W_1^{N-1},X_i)\}_{j\in T}.
%\end{align*}
%
%From (\ref{combine_2}) and (\ref{combine_4}), we get
%\begin{align*}
%R&=\{(R_N+\sum_{i\in T,i\neq j} R_i)+(R_j+R_u)\}_{j\in T}\\
%&<\{I(X_1,X_{T_j},U_{T_j};Y_N|X_j,U_j,W_1^{N-1})\\
%&~~~+I(U_j,W_1^{j-1};Y_j|X_j,W_j^{N-1})\}_{j\in T}.
%\end{align*}
%From (\ref{combine_2}), (\ref{combine_3}) and (\ref{combine_4}), we get
%\begin{align*}
%R&=(R_N+\sum_{i\in S} R_i)+\min_{j\in S^c}[(R_j+R_0)+(\sum_{i\in S^c,i\neq j} R_i)]\\
%&<I(X_1,X_{S},U_{S};Y_N|X_{S^c},U_{S^c},W_1^{N-1})\\
%&~~~+\min_{j\in S^c}[I(U_j,W_1^{j-1};Y_j|X_j,W_j^{N-1})\\
%&~~~+\sum_{i\in S^c,i\neq j} I(U_i;Y_i|W_1^{N-1},X_i)],
%\end{align*}
From (\ref{combine_2}), (\ref{combine_3}) and (\ref{combine_4}), we get
\begin{align}
\label{single_R_res1}
\nonumber
R&=\left(R_{N+1}+\sum_{i\in S} R_i\right)+\min_{j\in S^c}\left[\left(R_j+R_0\right)+\left(\sum_{i\in S^c,i\neq j} R_i\right)\right]\\
&<I(X_0,X_{S},U_{S};Y_{N+1}|X_{S^c},U_{S^c},W_0^{N})\\
\nonumber
&~~~+\min_{j\in S^c}I(W_0^{j-1};Y_j|X_j,W_j^{N})+\sum_{i\in S^c} I(U_i;Y_i|W_0^{N},X_i),
\end{align}
for all $S\subset T$.\par
From (\ref{combine_3}), (\ref{combine_4}), (\ref{combiniation_5}) and (\ref{combiniation_6}), we get
\begin{align}
\label{double_R_res2}
\nonumber
2R&<\sum_{i=1}^{N+1}R_i+R_{N+1}\\
\nonumber
&~~+\min_{i,j\in T}\left[ \left(R_i+R_0\right)+ \left(R_j+R_0\right)+\sum_{l \in T, l\neq i,j}
R_l\right]\\
\nonumber
&<I(U_1^{N},X_0^{N};Y_{N+1}|W_0^{N})+I(X_0;Y_{N+1}|U_1^{N},X_1^{N},W_0^{N})\\
\nonumber
&~~+2\min_{j\in T}I(W_0^{j-1};Y_j|X_j,W_j^{N})\\
&~~+2\sum_{i\in T} I(U_i;Y_i|W_0^{N},X_i),
\end{align}\par
However, if we let $S=\emptyset$ in (\ref{single_R_res1}) and double the right-hand-side (RHS) expression, then we can get a smaller expression than the RHS of (\ref{double_R_res2}). Thus, (\ref{double_R_res2}) is redundant.\par
After this combination process, we get the rate in (\ref{thm_pdf_DM-RN}).
\end{proof}
\subsection{Special Networks}
\begin{figure}[t]
\center
  \includegraphics[scale=0.75]{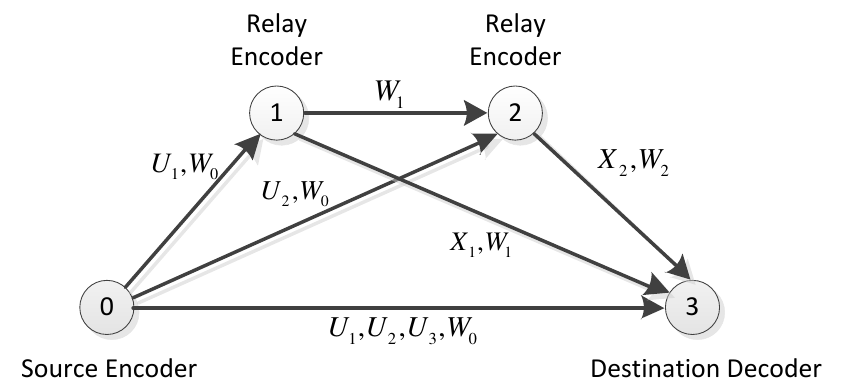}
  \caption{Two-level relay network with partial-decode-forward} \label{pdf_two_relay}
\end{figure}
\begin{figure}[t]
\center
  \includegraphics[scale=0.68]{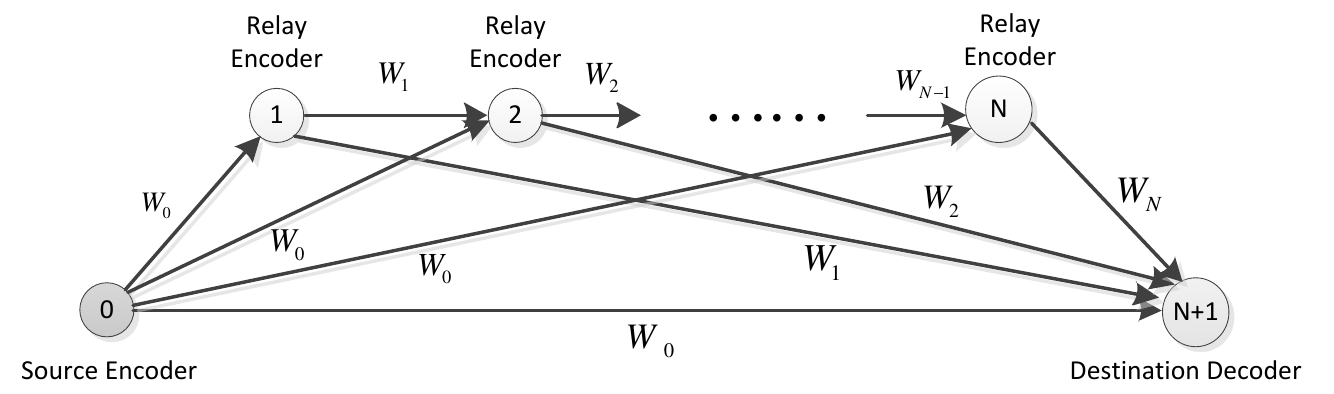}
  \caption{Decode-forward relay network} \label{DF_RN_1}
\end{figure}
If $N=2$, we have the partial decode-and-forward lower bound for a two-level relay network as shown in Figure \ref{pdf_two_relay}, which coincides with the result in \cite{aref2010pdftwolevel}.\par
For a general $N$, if we set private parts $U_1^{N}=0$ and $m_{N+1}=0$, we can get Xie and Kumar's \cite{Xie2005multilevelrelay} network decode-forward lower bound as shown in Figure \ref{DF_RN_1}. Furthermore, if $N=1$, it reduces to the decode-forward lower bound \cite{cover1979capacity} for the discrete-memoryless relay channel (DM-RC).
\section{Gaussian Relay Networks}\label{sec_gaussian_two_level}
The Gaussian relay network can be modeled as
\begin{align}
Y_k=\sum_{i=0}^{k-1}g_{ik}X_i+Z_k,
\end{align}
where $k\in T\cup\{N+1\}$, $g_{ik}$ is the coefficient of the link from node $i$ to node $j$, and $Z_k$ is noise at the decoder with Gaussian distribution $\mathcal{N}(0,1)$. There is a power constraint at each transmitting node as $P_k$.\par
As shown in Figure \ref{pdf_two_relay}, the Gaussian two-level relay network can be modeled as
\begin{align}
\label{gaussian_two_relay}
\nonumber
Y_3&=g_{03}X_0+g_{13}X_1+g_{23}X_2+Z_3,\\
Y_2&=g_{02}X_0+g_{12}X_1+Z_2, ~~~~Y_1=g_{01}X_0+Z_1,
\end{align}
where $Z_3$, $Z_2$ and $Z_1$ are independent AWGN noise according to the normal distribution $\mathcal{N}(0,1)$. The signaling at each node can be written as
\begin{align}
\label{gaussian_two_relay_signaling}
\nonumber
x_2&=\alpha_{22} W_2(w_{0,j-2})+\beta_{22} V_2(w_{2,j-2}),\\
\nonumber
x_1&= \alpha_{11} W_1(w_{0,j-1})+\alpha_{12} W_2(w_{0,j-2})+\beta_{11} V_1(w_{1,j-1}),\\
\nonumber
x_0&=\alpha_{00} W_0(w_{0,j})+ \alpha_{01} W_1 (w_{0,j-1})+ \alpha_{02} W_2(w_{0,j-2})\\
\nonumber
&~~~
+\beta_{01}V_1(w_{1,j-1})+\beta_{02}V_2(w_{2,j-2})\\
&~~~+\phi_{01} U_1(w_{1,j})+ \phi_{02} U_2(w_{2,j})+\phi_{03} U_3(w_{3,j}),
\end{align}
where $W_2$, $V_2$, $W_1$, $V_1$, $W_0$, $U_1$, $U_2$, $U_3$ are independent, normalized Gaussian random variables $\mathcal{N}(0,1)$; $\{\alpha_{*},\beta_{*},\phi_{*}\}$ are power allocations satisfying the following constraints:
\begin{align}
\nonumber
\label{power_allocation}
\alpha_{22}^2+\beta_{22}^2 = P_2,~~~~~~
\alpha_{11}^2+\alpha_{12}^2+\beta_{11}^2 &= P_1,\\
\!\!\!\!\!\!\!\!\alpha_{00}^2+\alpha_{01}^2+
 \alpha_{02}^2+\beta_{01}^2+\beta_{02}^2
 +\phi_{01}^2+\phi_{02}^2+\phi_{03}^2&=P_0,
\end{align}
where $P_0$, $P_1$ and $P_2$ are power constraints at the corresponding node, which can be set equal to each other without loss of generality.\par
\begin{cor}
The capacity for a Gaussian two-level relay network in (\ref{gaussian_two_relay}) is lower bounded by:
\begin{align}
\label{cor:two_level_gaussian}
C\geq\min \left\{
I_1+I_4+I_5,
I_2+I_3+I_5,
I_2+I_7,
I_4+I_6,
I_8 \right\},
\end{align}
where
\begin{align*}
I_1&= \frac{1}{2}\log\left(1+\frac{g_{01}^2\phi_{01}^2}{g_{01}^2(\beta_{02}^2+\phi_{02}^2+\phi_{03}^2)+1}\right),\\
I_4&= \frac{1}{2}\log\left(1+\frac{g_{02}^2(\alpha_{00}+\phi_{02})^2+(g_{02}\alpha_{01}+g_{12}\alpha_{11})^2}{(g_{02}\beta_{01}+g_{12}\beta_{11})^2+g_{02}^2(\phi_{01}^2+\phi_{03}^2)+1}\right),\\
I_3&= \frac{1}{2}\log\left(1+\frac{g_{02}^2\phi_{02}^2}{(g_{02}\beta_{01}+g_{12}\beta_{11})^2+g_{02}^2(\phi_{01}^2+\phi_{03}^2)+1}\right),\\
I_2&= \frac{1}{2}\log\left(1+\frac{g_{01}^2(\alpha_{00}^2+\phi_{01}^2)}{g_{01}^2(\beta_{02}^2+\phi_{02}^2+\phi_{03}^2)+1}\right),\\
I_5&= \frac{1}{2}\log\left(1+g_{03}^2\phi_{03}^2\right),\\
I_6&= \frac{1}{2}\log\left(1+(g_{03}\beta_{01}+g_{13}\beta_{11})^2+g_{03}^2(\phi_{01}^2+\phi_{03}^2)\right),\\
I_7&= \frac{1}{2}\log\left(1+(g_{03}\beta_{02}+g_{23}\beta_{22})^2+g_{03}^2(\phi_{02}^2+\phi_{03}^2)\right),\\
I_8&= \frac{1}{2}\log\left(1+g_{03}^2P_0+g_{13}^2P_1+g_{23}^2P_2\right.\\
&~~~~+2g_{03}g_{13}(\alpha_{01}\alpha_{11}+\alpha_{02}\alpha_{12}+\beta_{01}\beta_{11})\\
&~~~~\left.+2g_{03}g_{23}(\alpha_{02}\alpha_{22}+\beta_{02}\beta_{22})+2g_{13}g_{23}\alpha_{12}\alpha_{22}\right),
\end{align*}
 and $\alpha_{ij},\beta_{ij},\phi_{ij}$ ($i\in\{0,1,2\}$,$j\in\{0,1,2,3\}$) are power allocations satisfying (\ref{power_allocation}).
 \end{cor}
 \begin{proof}
 Applying the signaling in \eqref{gaussian_two_relay_signaling} to the rate region in Theorem \ref{thm_pdf_DM-RN}, we obtain \eqref{cor:two_level_gaussian}.
 \end{proof}
\section{Conclusion}\label{sec_concluding}
In this paper, we consider partial decode-forward relaying in a single-source single-destination network with $N$ relays. We design a scheme in which each relay forwards the common message part and a specific private part to the following nodes. The proof is based on block Markov encoding and simultaneous sliding window decoding. The key point is that each relay decodes and forwards its private part only when the last common part with the same block index arrives. We then obtain the achievable rate for this scheme, which can be expressed in a compact form over all cutsets and permutations of relays. We show that this scheme contains existing results for an $N$-relay network with decode-forward and a two-level relay network with partial decode-forward considering all message splitting cases. We then study the Gaussian two-level relay network and derive the achievable rate by the proposed scheme.
\appendices
\section{Error analysis at relay $k$}
\label{error_analysis_relay}
\begin{figure*}
\begin{align}
\nonumber
P_e\leq& P_c(E^c(\{1\}_{N+1},1))+\sum_{m_{0,j-N}\neq 1,\{m_{i,j-N}\}_{i\in T}, m_{N+1,j-N}}P_c(E(\{m_{i,j-N}\}_{i\in T},m_{N+1,j-N}, m_{0,j-N}))\\
%%another group
&+\sum_{m_{0,j-N}=1,\{m_{i,j-N}\}_{i\in S}\neq1,\{m_{i,j-N+2}\}_{i\in S^c}=1, m_{N+1,j-N}} P_c(E(\{m_{i,j-N}\}_{i\in S},\{1\}_{|S^c|},m_{N+1,j-N}))
\label{union_bound_destination}
\end{align}
\end{figure*}
%\begin{figure*}
%\begin{align}
%\nonumber
%&\sum_{\{m_{i,j-N}\}_{i\in S}\neq1,\{m_{i,j-N}\}_{i\in S^c}=1,  m_{N+1,j-N}, m_{0,j-N}=1}P(E(\{m_{i,j-N}\}_{i\in S},\{1\}_{|S^c|},m_{N+1,j-N}, 1))\\
%\nonumber
%&=P(\cup_{\{m_{i,j-N}\}_{i\in S}\neq1}\cup_{m_{N+1,j-N}\neq1}(\cap_{i\in S}E_{i}\cap E_{N+1}))
%\leq \sum_{\{m_{i,j-N}\}_{i\in S}\neq1}\sum_{m_{N+1,j-N}\neq1}\prod_{i\in S}P(E_{i})\times P(E_{N+1})\\
%\nonumber
%&\leq\prod_{i\in S}\{2^{nR_i}\}\times2^{nR_{N+1}}\times\prod_{i\in S}2^{-n(I(X_i;Y_{N+1}|W_i^{N},X_{i+1}^{N})-\delta(\epsilon))}\times 2^{-n(I(X_0,U_S;Y_{N+1}|W_0^{N},X_{1}^{N},U_{S^c})-\delta(\epsilon))}\\
%\label{destination_error_analysis_2}
%&\leq\prod_{i\in S}\{2^{nR_i}\}\times2^{nR_{N+1}}\times
%2^{-n(I(X_0,X_S,U_S;Y_{N+1}|X_{S^c},U_{S^c},W_0^{N})
%-(|S|+1)\delta(\epsilon))},
%\end{align}
%\end{figure*}
Assume without loss of generality that $(m_{k,j-k+1}, m_{0,j-k+1})=(1,1)$ is sent in block $j$. \par
We first define the following events: \par
\begin{itemize}
\item
$E_i(m_{k,j-k+1}, m_{0,j-k+1})$, for $i\in[1:k]$, is when only the $i$th decoding rule is satisfied. We simplify the notation as $E_i$ in the following analysis.
\item
$E(m_{k,j-k+1}, m_{0,j-k+1})$ as the event that all decoding rules are satisfied simultaneously.
\end{itemize}\par
Then, by the union bounds, the probability of error is bounded as
\begin{align*}
P_e&\leq P_c(E^c(1,1))\\
&~~~+\sum_{m_{0,j-k+1}=1,m_{k,j-k+1}\neq1 }P_c(E(m_{k,j-k+1}, 1))\\
&~~~+\sum_{m_{0,j-k+1}\neq1,m_{k,j-k+1} }P_c(E(m_{k,j-k+1}, m_{0,j-k+1})),
\end{align*}
where $P_c$ is the conditional probability given that $(1,1)$ was sent.\par
By the law of large number, $P_c(E^c(1,1))\to0$ as $n\to \infty$.\par
By the joint typicality lemma, we have
\begin{align*}
\sum_{ m_{0,j-k+1}=1,m_{k,j-k+1}\neq1}P_c(E(m_{k,j-k+1}, 1))\\
\leq 2^{nR_k}\times2^{-n(I(U_k;Y_k|W_0^{N},X_k)-\delta(\epsilon))},
\end{align*}\par
which goes to 0 as $n\to \infty$, if
\begin{align*}
R_k<I(U_k;Y_k|W_0^{N},X_k)-\delta(\epsilon).
\end{align*}\par
According to the independence of the codebooks and the joint typicality lemma,
\begin{align*}
&\sum_{m_{0,j-k+1}\neq1,m_{k,j-k+1} }P_c(E(m_{k,j-k+1}, m_{0,j-k+1}))\\
&=P_c(\cup_{ m_{k,j-k+1}}\cup_{m_{0,j-k+1\neq1}}(E_1\cap E_2\cap\dots\cap E_{k}))\\
&\leq \sum_{m_{k,j-k+1}}\sum_{m_{0,j-k+1\neq1}}P(E_1)\times P(E_2)\times\dots\times P(E_{k})\\
&\leq 2^{nR_k} \times 2^{nR_0} \times
2^{-n(I(W_{k-1};Y_k|X_k,W_k^{N})
-\delta(\epsilon))}\\
&~~~\times 2^{-n(I(W_{k-2};Y_k|X_k,W_{k-1}^{N})
-\delta(\epsilon))}\times\dots\times\\
&~~~2^{-n(I(W_2;Y_k|X_k,W_3^{N})
-\delta(\epsilon))}
\times
2^{-n(I(U_k,W_1;Y_k|X_k,W_2^{N})
-\delta(\epsilon))},
\end{align*}
which tends to 0 as $n\to \infty$ if
\begin{align*}
R_k+R_0<I(U_k,W_1^{k-1};Y_k|X_k,W_k^{N})-k\delta(\epsilon).
\end{align*}

\section{Error analysis at destination $N+1$}
\label{error_analysis_destination}

Assume without loss of generality that $(\{m_{k,j-N}\}_{k\in T}, m_{0,j-N},m_{N+1,j-N})=(1,1,\dots,1)$ was sent in block $j$. \par
We first define the following events:
\begin{itemize}
\item
$E_i(\{m_{k,j-N}\}_{k\in T}, m_{0,j-N},m_{N+1,j-N})$, $i\in [1:N+1]$, as the event that only the $i$th decoding rule is satisfied. We simplify the notation as $E_i$ in the following analysis.
\item
$E(\{m_{k,j-N}\}_{k\in T}, m_{0,j-N},m_{N+1,j-N})$ as the event that all $N+1$ decoding rules are satisfied simultaneously.
\end{itemize}\par
We define set $S$ to be the set of wrongly decoded private messages and $S^c$ as the set of correctly decoded private messages at the destination. Then, by the union of events bound, the probability of error is bounded as in (\ref{union_bound_destination}), where $P_c$ is the conditional probability given that $(\{m_{k,j-N}\}_{k\in T}, m_{0,j-N},m_{N+1,j-N})=(1,1,\dots,1)$ was sent.
By the law of large number, $P_c(E^c(\{1\}_{N+1},1))\to0$ as $n\to \infty$. It is impossible to correctly decode  $m_{N+1,j-N}$ if any of $\{m_{k,j-N}\}_{k\in T}$ isn't decoded correctly.
According to the independence of codebooks and the joint typicality lemma, the third term in (\ref{union_bound_destination})  tends to 0 as $n\to \infty$ if
\begin{align*}
\sum_{i\in S} R_i+R_{N+1}<&I(X_0,\{X_i,U_i\}_{i\in S};Y_{N+1}|\{X_j,U_j\}_{j\in S^c},W_0^{N})\\
&-(|S|+1)\delta(\epsilon).
\end{align*}\par
Similarly, according to the independence of codebooks and the joint typicality lemma, the second term in (\ref{union_bound_destination}) tends to 0 as $n\to \infty$ if
\begin{align*}
\sum_{i=1}^{N+1} R_i+R_0<I(U_1^{N},X_0^{N},W_0^{N};Y_{N+1})-(N+1)\delta(\epsilon).
\end{align*}
\bibliographystyle{IEEEtran}
\bibliography{reflist}

\end{document}